\documentclass{amsart}
\usepackage[utf8]{inputenc}
\usepackage{cite}
\usepackage{amsmath,amssymb,amsfonts}
\usepackage{algorithmic}
\usepackage{graphicx}
\usepackage{textcomp}
\usepackage{xcolor}
\usepackage{lipsum}
\usepackage{mathtools}
\usepackage[foot]{amsaddr}
\usepackage{dsfont}
\usepackage{comment}
\newtheorem{proposition}{Proposition}
\newtheorem{definition}{Definition}

\newtheorem{theorem}{Theorem}
\newtheorem{lemma}{Lemma}

\newtheorem{conjecture}{Conjecture}
\newtheorem{remark}{Remark}
\DeclarePairedDelimiter\ceil{\lceil}{\rceil}
 \DeclareMathOperator{\diag}{diag}

\title{Robust Sparse Recovery with Sparse Bernoulli Matrices via Expanders}
\author[1]{Pedro Abdalla}

\email{pedro.abdallateixeira@ifor.math.ethz.ch}

\address{Department of Mathematics, ETH Z\"urich}

\date{December 2021}

\begin{document}

\maketitle


\begin{abstract}
Sparse binary matrices are of great interest in the field of compressed sensing. This class of matrices make possible to perform signal recovery with lower storage costs and faster decoding algorithms. In particular, random matrices formed by i.i.d Bernoulli $p$ random variables are of practical relevance in the context of nonnegative sparse recovery.

In this work, we investigate the robust nullspace property of sparse Bernoulli $p$ matrices. Previous results in the literature establish that such matrices can accurately recover $n$-dimensional $s$-sparse vectors with $m=O\left (\frac{s}{c(p)}\log\frac{en}{s}\right )$ measurements, where $c(p) \le p$ is a constant that only depends on the parameter $p$. These results suggest that, when $p$ vanishes, the sparse Bernoulli matrix requires considerably more measurements than the minimal necessary achieved by the standard isotropic subgaussian designs. We show that this is not true. Our main result characterizes, for a wide range of sparsity levels $s$, the smallest $p$ such that it is possible to perform sparse recovery with the minimal number of measurements. We also provide matching lower bounds to establish the optimality of our results.
\end{abstract}


\section{Introduction}
The theory of compressed sensing, introduced in the seminal works \cite{candes2006robust,candes2006stable,donoho2006compressed}, predicts that a few samples are sufficient to recover a sparse vector from linear corrupted measurements. In a mathematical framework, the goal is to recover an $s$-sparse vector $x \in \mathbb{R}^n$ from linear corrupted measurements $y=Ax +e \in \mathbb{R}^m$, where $A\in \mathbb{R}^{m\times n}$ is a known matrix often called "measurement" matrix and $e$ is a noisy vector satisfying $\|e\|_2 \le \eta$, for some constant $\eta >0$. The number of rows $m$ is referred as the number of measurements and it is often assumed to be much smaller than the ambient dimension, i.e, $m \ll n$.

It is well established in the literature that $m=\Theta(s\log\frac{en}{s})$ measurements are sufficient and necessary to recover an $s$-sparse vector $x \in \mathbb{R}^n$ \cite{foucart2013invitation}. Precisely, when the measurement matrix $A$ fulfill appropriate conditions, for example a robust version of the nullspace property (NSP), the target sparse vector can be recovered via the so called quadratically constrained basis pursuit (QCBP), the optimization program 

\begin{equation}
\label{QCBP}
\hat{x}= \arg \min_{z}  \|z\|_1      \quad   \text{subject to}\quad \|Az-y\|_2\leq \eta.
\end{equation}

To date, every measurement matrix that requires the minimal number of measurements, $m=\Theta(s\log\frac{en}{s})$, to perform sparse recovery is generated at random. It is well known that matrices with i.i.d isotropic subgaussian rows satisfy the so called $\ell_2$ restricted isometry property (RIP) with high probability, a sufficient, but not necessary, condition to establish recovery guarantees via the optimization program \eqref{QCBP} \cite{foucart2013invitation}. The assumption that the rows are subgaussian was significantly relaxed to mild moment assumptions in a line of work based on the Mendelson's small ball method \cite{lecue2017sparse,dirksen2016gap,tropp2015convex}. Another line of work, focused on additional desired properties of the randomly generated measurement matrices such as sparsity. In this sense, sparse binary matrices were considered and random matrices with i.i.d  entries following a Bernoulli $p$ distribution would be natural model to generate sparse binary matrices. The first result in this direction is due to Chandar \cite{chandar2008negative}, who proved that any binary matrix cannot satisfy the RIP property with the number of rows proportional to $s\log(\frac{en}{s})$. However, Berinde, Gilbert, Indyk, Karloff and Strauss \cite{berinde2008combining} introduced a class of sparse measurement matrices that perform sparse recovery with $m=\Theta(s\log\frac{en}{s})$ measurements via the so called lossless expanders. In a nutshell, these are left regular bipartite graphs with adjacency matrix that satisfies an $\ell_1$ version of the restricted isometry property. Moreover, it is possible to prove that a random left $d$-regular bipartite graph, for an appropriate value of $d$, is a lossless expander with high probability, see \cite{foucart2013invitation} and the references therein. Nachin \cite{nachin2021lower} showed that under the assumption that $m=\Theta(s\log\frac{en}{s})$, such matrices achieve the minimum possible column sparsity (number of ones per column) up to an absolute constant.
In a similar spirit, Prasad and Rudelson \cite{kasiviswanathan2019restricted} investigated the restricted isometry property of sparse subgaussian matrices, i.e, matrices generated by the Hadamard product (entrywise product) between a random matrix with mean zero isotropic subgaussian entries and a Bernoulli $p$ matrix. They obtained that $m=\Theta(\frac{s}{p}\log\frac{en}{s})$ measurements are enough to establish the RIP property and consequently the NSP property.

\subsection{Nonnegative Compressed Sensing}
The line of work in nonnegative compressed sensing aims to study recovery guarantees under the assumption that the target vector is nonnegative or the measurement matrix is nonnegative. Donoho and Tanner \cite{donoho2005sparse} studied the central question: Which assumptions on the measurement matrix $A$ are necessary to ensure the uniqueness of sparse solution under the nonnegative assumption? Mathematically we need,
\begin{equation*}
    \{z| Ax=Az \ \& \  z\ge 0\} = \{x\}.
\end{equation*}
To answer this question, the authors introduced the notion of outwardly $s$-neighborly polytopes: Consider $P_{A}$ the polytope formed by the convex hull of the columns of $A$, the polytope $P_A$ is called $s$-neighborly if every set of $s$ vertices spans a face of $P_A$. If it happens, the polytope $P_A^{0}:=P_A\cup \{0\}$ is called outwardly $s$-neighborly. They proved that such condition is necessary and sufficient to establish perfect recovery in the noiseless case. Another approach to the same question is due to Bruckstein, Elad and Zibulevsky \cite{bruckstein2008uniqueness}. They considered a full rank measurement matrix $A$ and introduced the notion of positive orthant condition. 
\begin{equation*}
    \mathcal{M}^{+}:= \{A: \exists t \in \mathbb{R}^m \ \text{such that} \ A^{T}t >0\}.
\end{equation*}
The authors analysed the performance of \eqref{QCBP} with the additional constraint $x\ge 0$ (all entries of $x$ are nonnegative) via coherence type analysis. They obtained numerical evidence for unique recovery for $s=O(\sqrt{n})$. Wang, Xu and Tang \cite{wang2010unique}, motivated by applications in network inference problems, analysed the signal recovery problem with binary matrices formed by i.i.d Bernoulli entries via a RIP-based technique. Precisely, they considered a matrix $A_1 \in \{0,1\}^{m\times n}$ formed by i.i.d Bernoulli $\frac{1}{2}$ entries concatenated with a row of all ones $\mathbf{1}_n$, so their measurement matrix is $A^{T}:=[\mathbf{1}_n^T|A_1^T] \in \{0,1\}^{(m+1)\times n}$. Observe that every vector in the kernel of $A$ is orthogonal to all ones vector $\mathbf{1}_n$. Besides, it is easy to see that every vector orthogonal to all ones vector $\mathbf{1}_n$ belongs to the kernel of $A_1$ if and only if it belongs to the kernel of $B:=2A_1 - J$, where $J \in \{1\}^{m\times n}$ (all ones matrix). We have just obtained that $Ker(A) \subset Ker(B)$. Since the entries of $B$ are i.i.d random signs, $B$ satisfies the standard RIP property and consequently the NSP property is transferred from the matrix $B$ to the matrix $A$. The downside of this approach is that it only provides guarantees for the noiseless setting and the measurement matrix is not really sparse. Kueng and Jung in \cite{kueng2017robust}, motivated by concrete applications in wireless network detection, analysed the measurement matrices formed by i.i.d entries Bernoulli $p$ using Mendelson's small ball method. In order to describe their results, we first introduce the main theorem related to the small ball method.

\begin{proposition}{\cite{tropp2015convex,MR3367000}}
Fix a set $E \subset \mathbb{R}^n$ and let $\mathbf{a}_1,\ldots,\mathbf{a}_m$ be i.i.d copies of the random vector $\mathbf{a} \in \mathbb{R}^n$. Set $\mathbf{h}:= \frac{1}{\sqrt{m}}\sum_{i=1}^m \varepsilon_i \mathbf{a}_i $, where $\varepsilon_1,\ldots,\varepsilon_m$ are i.i.d Rademacher (random signs) random variables. For $\xi>0$, define 
\begin{equation*}
    Q_{\xi}(E,\mathbf{a}) = \inf_{\mathbf{u}\in E} \mathbb{P}(|\langle \mathbf{a},\mathbf{u} \rangle | \ge \xi) \quad \text{and} \quad W_m(E,\mathbf{a}) = \mathbb{E} \sup_{\mathbf{u} \in E} \langle \mathbf{h},\mathbf{u} \rangle.
\end{equation*}
Then, for any $\xi >0$ and $t\ge 0$, with probability at least $1-e^{-2t^2}$,
\begin{equation*}
    \inf_{\mathbf{v} \in E} \left (\sum_{i=1}^m |\langle \mathbf{a}_k,\mathbf{v}\rangle|^2 \right)^{1/2} \ge \xi \sqrt{m}Q_{2\xi}(E,\mathbf{a}) - \xi t - 2W_m(E,\mathbf{a}).
\end{equation*}
\end{proposition}

The standard application of the small ball method in compressed sensing goes as follows: Choose $E=\{\mathbf{v} \in S^{n-1}: \exists S\subset [n], |S|\le s \ \text{s.t} \  \|v_S\|_{2} > \frac{\rho}{\sqrt{s}}\|v_{S^c}\|_1\}$ with $\rho <1$ being a constant. It follows that, for $m = \Omega (W_m(E,\mathbf{a})^2/Q_{2\xi}(E,\mathbf{a})^2)$, with high probability, no vector $\mathbf{v}\in E$ lies in the kernel of $A$, therefore $A$ satisfy the $\ell_2$ robust nullspace property. It remains to choose a suitable distribution for the random vector $\mathbf{a}$ such that $Q_{\xi}(E,\mathbf{a}) >0$ and $W_m(E,\mathbf{a})/Q_{2\xi}(E,\mathbf{a}) \sim s\log\frac{en}{s}$. It was shown by Lecue and Mendelson \cite{lecue2017sparse} that isotropic random vectors, with $\log(n)$ bounded moments satisfy these requirements.

Kueng and Jung applied the small method with $\mathbf{a} \in \mathbb{R}^n$ being a random vector with i.i.d Bernoulli $p$ entries. They obtained that, with high probability, for $m=\Theta(\frac{s(2p-1)}{p^3(1-p)^3\log \frac{p}{1-p}}\log\frac{en}{s})$, the matrix $A \in \{0,1\}^{m\times n}$ formed by i.i.d entries Bernoulli $p$ satisfy the $\ell_2$ robust null space property of order $s$. The proof mimes the techniques used in \cite{lecue2017sparse,dirksen2016gap}. They also showed that $A$ belongs to the set $\mathcal{M}^{+}$ with high probability.

Combining these two facts together, Kueng and Jung were able to show that, if we assume that the vector to be recovered is nonnegative and $s$-sparse, then it is possible to prove recovery guarantees for a simple nonnegative least squares:

\begin{equation}
\hat{x}= \arg \min_{z\ge 0}  \|Az - y\|_2.
\end{equation}

The constraint $z\ge 0$ means that all entries of $z$ are nonnegative. Remarkably, the algorithm does not use any knowledge of the noise level $\eta$ and it gives rise to significant computational savings because the least squares algorithm is considerably cheaper than the QCBP \eqref{QCBP} in terms of computational costs.

For this work, the main question raised by Kueng and Jung is the following: What is the correct behaviour of the number of measurements in the practically relevant regime when $p \rightarrow 0$? \cite[Remark 10]{kueng2017robust}.

We first observe that an important downside of the small ball method is the necessity to have $Q_{2\xi}(E,\mathbf{a})>0$, i.e, the random vector $\mathbf{a}$ must satisfy a small ball condition, otherwise the conclusion of the theorem is empty, the right hand side becomes negative. The small ball condition has been analysed in a recent line of work dedicated to remove such condition in many different problems in mathematical statistics and mathematical signal processing \cite{saumard2018optimality,mourtada2021distribution,vavskevivcius2020suboptimality,krahmer2017phase}. Indeed, the estimates provided by the small ball method quickly deteriorates when $Q_{2\xi}(E,\mathbf{a})$ is not an absolute constant. Arguably, random matrices formed by Bernoulli $p$ entries with $p=o(1)$ are the most important example not captured by the small ball method. It is easy to see that $Q_{2\xi}(E,\mathbf{a})\le p$ when the vector $\mathbf{a}$ is formed by i.i.d for Bernoulli $p$ entries. It was pointed out by Kueng and Jung \cite[Section B]{kueng2017robust} that the small ball method cannot give sharper dependency on $p$ up to an absolute constant.

Recently, Jeong, Li, Plan and Yilmaz \cite{jeong2020sub} improved the dependency on $p$ in the result of Kueng and Jung as an application of their refined concentration inequalities. They obtained standard concentration inequalities, for example Bernstein inequality, with a better dependency on the so called Orlicz $\psi_2$ norm of a random variable. We include here the definition for the sake of completeness,
\begin{definition}
A random variable $X$ is subgaussian if the following norm is finite
\begin{equation*}
    \|X\|_{\psi_2}:= \inf\{t>0: \mathbb{E}[e^{-X^2/t^2}]\le 2\}.
\end{equation*}
\end{definition}
It can be shown that the quantity above is indeed a norm \cite{vershynin2018high}. What is important here is the fact that, for a Bernoulli $p$ random variable, the $\psi_2$ norm depends on $p$, so it is important that the concentration inequality to be used scales correctly with the $\psi_2$ norm. If the parameter $p$ is an absolute constant, then the deterioration of such norm is not relevant.

Formally, they proved that, with high probability, $m = \Theta(\frac{s}{p(1-p)}\log(\frac{en}{s}))$ measurements are enough to establish the $\ell_2$ robust null space property of order $s$. The proof is based on a sophisticated RIP-type analysis. They normalized a Bernoulli $p$ matrix $A$ to become isotropic, $\Tilde{A}: = \frac{1}{\sqrt{p(1-p)}}(A- \mathbb{E}A)$ and considered a projection $P\in \mathbb{R}^{m\times m}$ onto the vector space orthogonal to the all ones vector. Then they proved, via a refined matrix deviation inequality, that $P\Tilde{A} = PA$ satisfies the RIP property and consequently the $\ell_2$ robust NSP property. Finally, they translated the latter property property to $A$ by observing that $Ker(A) \subset Ker(PA)$. We remark that the authors in \cite{jeong2020sub} also proved that $mp>\frac{1}{2}$ is necessary and they explicitly wrote that the dependence on $p$ should be sharp up to absolute constants. It turns out this is far from true (see the discussion after Theorem \ref{thm:Main}). In this paper, we address the following question:
\\

(Q): In the regime $m\ll n$, let $A \in \{0,1\}^{m\times n}$ be a random matrix formed by i.i.d entries Bernoulli $p$. What is the smallest order of $p$ such that the matrix $A$ satisfies a robust nullspace property of order $s$ with $m=C s\log \frac{en}{s}$ rows, where $C>0$ is an absolute constant?

\subsection{Main Results}
Our main result provides a sharp answer for the question raised (Q) for a wide range of sparsity levels $s$. The main theorem is the following:
\begin{theorem}{(Main Theorem)}
\label{thm:Main}
Let $A\in \{0,1\}^{m\times n}$ be a random matrix formed by i.i.d entries following a Bernoulli $p$ distribution. Then there exists absolute constants $c_1,c_2,c_3,c_4,c_5,c_6>0$ with $c_1,c_4< 1$ such that for $p\le \frac{c_1}{s}$, with probability at least $\frac{c_2}{n^{c_3-1}}$, $A$ satisfies the $\ell_1$ robust nullspace property of order $s$ with parameters $\rho=c_4$ and $\tau =\frac{1}{c_5pm}$, with respect to the $\ell_1$ norm provided that $m\ge c_6c_3\frac{\log n}{p}$.

In particular, for $p=\frac{c_1}{s}$, $m = \frac{1}{c_1}c_3c_6s\log n$ measurements with $c_3>1$ are enough to guarantee the robust nullspace property with high probability.
\end{theorem}

In Section \ref{sec:lowerbounds}, we show that $m=\Omega(\frac{\log n}{p})$ is necessary, therefore the main Theorem \ref{thm:Main} cannot be improved, it can only be extended for larger values of $p$. We also remark that the regime $s\ll n$ was considered due to practical applications such as wireless network \cite{kueng2017robust}. Under a slightly stronger assumption, $s=O(n^{\beta})$ for some constant $0<\beta<1$, the functions $\log n$ and $\log(\frac{en}{s})$ are of the same order and then our result completely settles the question (Q). Moreover, due to concentration, each column of $A$ has at most $O(\log n)$ ones. This is sharp for $s=O(n^{\beta})$ and almost sharp for $n^{\beta}\ll s \ll n$ (up to an iterated log factor) in the following sense: If an $m\times n$ binary matrix with $m =\Theta(s\log(\frac{en}{s}))$ satisfies the nullspace property then each column has at least $\log\frac{en}{s}$ ones \cite{nachin2021lower}. Moreover, notice that all previous results scale with $m=\Omega(s^2\log\frac{en}{s})$ in the same sparse regime, the random constructions becomes worse than the standard deterministic constructions \cite{foucart2013invitation}. In particular, our main theorem disproves the claim in \cite{jeong2020sub} that $m=\Theta (\frac{s}{p}\log(\frac{en}{s}) )$ is optimal. Indeed, for simplicity, in the noiseless case, if the result mentioned were optimal then for $p=O(\frac{1}{s})$ and $s=O(n^{1/2})$ (say) we would need $m=\Omega (s^2\log \frac{en}{s})$ measurements to perform exact recovery, however Theorem \ref{thm:Main} guarantees that $m=O(s\log\frac{en}{s})$ measurements are enough to perform exact recovery in the same regime of $p$. The result also reveals a new phase transition not captured by the previous results. Finally, we also present a conjecture in Section \ref{sec:lowerbounds} about the precise answer for the question (Q) in full generality.

As an application, we derive noise blind guarantees for nonnegative compressed sensing with sparse Bernoulli matrices. In a few words, a simple $\ell_1$ minimization algorithm suffices to perform signal recovery. No knowledge of the noise level is required. After publishing this manuscript, the author was notified that a similar result was obtained with a different measurement matrix, see \cite{petersen} for more details. We refer the reader to Section \ref{sec:applications} for the formal statement and discussion.

\subsection{Organisation of the paper}
The rest of this paper is organised as follows: In Section \ref{sec:preliminaries+background} we provide some background results. In Section \ref{sec:quasi-regular_expanders} we provide the main ideas of this work and the proof of Theorem \ref{thm:Main}. In Section \ref{sec:lowerbounds} we establish lower bounds and the precise conjecture. Section \ref{sec:applications} is dedicated to applications of our main results to noise blind compressed sensing. The appendix is dedicated to technical proofs.

\section{Preliminaries \& Background}
\label{sec:preliminaries+background}
We start by introducing some notation. The set $\{1,\ldots,n\}$ is denoted by $[n]$. For $1\le q\le \infty$, we write $\|.\|_q$ for the standard $\ell_q$ norm for vectors and $B_q^{n},S_q^{n-1}$ for the $n$-dimension unit ball and $n-1$-dimensional unit sphere with respect to the $\ell_q$ norm, respectively. For the standard Euclidean sphere $S_2^{n-1}$, we just write $S^{n-1}$. For functions $f(s,m,n)$ and $g(s,m,n)$ we write $f(s,m,n)\lesssim g(s,m,n)$ if there exists an absolute constant $C>0$ such that $f(s,m,n) \le C g(s,m,n)$, the notation $f(s,m,n)\gtrsim g(s,m,n)$ is defined analogously and we write $f(s,m,n) \sim g(s,m,n)$ if $f(s,m,n) \lesssim g(s,m,n)$ and $f(s,m,n)\gtrsim g(s,m,n)$. For a set $S \subset [n]$ we denote its complement over $[n]$ by $S^{c}$ and its cardinality by $|S|$. For a vector $x\in \mathbb{R}^n$, $x_S$ is the vector such that $(x_{S})_j = x_j$ for all $j\in S$ and $(x_{S})_j=0$ otherwise. The notation $x\ge 0$ means that the coordinates is nonnegative and $\mathbf{1}_m$ denotes all ones vector in $\mathbb{R}^m$. The best $s$-term approximation of the vector $x$ with respect to the $\ell_1$ norm is defined as $\sigma_s(x)_1:=\inf\{\|z-x\|_1: z \ \text{is $s$-sparse}\}$. For matrices $W \in \mathbb{R}^{n\times n}$, $\|W\|$ denotes the standard operator norm and $\diag(w)$ denotes the diagonal matrix formed by the vector $w \in \mathbb{R}^n$. A Bernoulli $p$ matrix $A$ is a random matrix formed by i.i.d entries $0/1$ Bernoulli $p$ random variables, i.e, random variables that are one with probability $p$ and zero otherwise. We now provide some background definitions and results.

\begin{remark}
We are interested in the regime when $p$ vanishes, therefore, for the rest of this paper, we explicitly assume that $p\le\frac{1}{2}$. 
\end{remark}
\begin{definition}{\cite{foucart2013invitation}}
Given $q\ge 1$, the matrix $A \in \mathbb{R}^{m\times n}$ satisfies the $\ell_q$ robust nullspace property of order $s$ with respect to a norm $\|.\|$ with constants $0<\rho<1$ and $\tau >0$ if, for any set $S \subset [n]$ with cardinality $|S|\le s$ and vector $v \in S^{n-1}$,
\begin{equation*}
  \|v_{S}\|_q \le \frac{\rho}{s^{1-1/q}} \|v_{S^c}\|_1 + \tau \|Av\|
\end{equation*}
\end{definition}
The next theorem is a classical result that relates nullspace property with recovery guarantees.
\begin{theorem}{\cite{foucart2013invitation}}
\label{Thm:4.25-Fourcart/Rauhut}
Suppose that the matrix $A \in \mathbb{R}^{m\times n}$ satisfies the $\ell_q$ robust null space property of order $s$ with constants $0<\rho<1$ and $\tau > 0$ with respect to a norm $\|.\|$. Then for any vectors $x,z \in \mathbb{R}^n$ and $1\le p\le q$,
\begin{equation*}
    \|z-x\|_p \le \frac{C}{s^{1-1/p}}(\|z\|_1 - \|x\|_1 + 2\sigma_s(x)_1) + D\tau s^{1/p-1/q} \|A(z-x)\|.
\end{equation*}
Here $C,D>0$ are constants depending only on $\rho$. In particular, if $\hat{x}$ is the solution of the optimization program \eqref{QCBP} then
\begin{equation*}
    \|\hat{x}-x\|_p \le \frac{C}{s^{1-1/p}}2\sigma_s(x)_1 + 2D\tau s^{1/p-1/q}\eta.
\end{equation*}
\end{theorem}

In a nutshell, our approach to tackle the question (Q) is based on spectral graph theory. The key idea is to view the Bernoulli $p$ matrix, for a certain regime of $p$, as the adjacency matrix of a bipartite graph that behaves similarly to a lossless expander. In order to provide a precise description of this idea, we present some background results in this direction. Consider a bipartite graph $G=(L,R,E)$, where the set of left nodes $L$ should be identified with $[n]$, the set of right nodes $R$ should be identified with $[m]$ and each edge $e \in \overline{ji} \in E$ connects a vertex $j\in L$ to other vertex $i\in R$. A graph $G=G(L,R,E)$ is said to be $d$-left regular if all nodes in $L$ have degree $d$.

\begin{definition}{(Lossless expanders)}
A $d$-left regular bipartite graph is an $(s,d,\theta)$-lossless expander if, for all subsets $J\subset L$ with $|J|\le s$, the set
\begin{equation*}
    R(J):= \{i \in R: \text{there is $j\in J$ with $\overline{ji}\in E$}\},
\end{equation*}
satisfies the following expansion property
\begin{equation*}
    |R(J)| \ge (1-\theta)d |J|.
\end{equation*}
The smallest $\theta>0$ that satisfies the above property is denoted by $\theta_s$.
\end{definition}
A natural concept associated with a graph is its adjacency matrix, the next definition and theorem establish the connection between the lossless expanders and compressed sensing.
\begin{definition}{(Adjacency matrix)}

The adjacency matrix of a bipartite graph $G=G(L,R,E)$ with $|L|=n$ and $|R|=m$ is the $m\times n$ matrix defined as follows
\begin{equation}
\label{def:adjancency}
    M_{i,j}:= \begin{cases}
    1, \text{if $\overline{ji} \in E$}\\
    0, \text{otherwise.}
    \end{cases}
\end{equation}
\end{definition}
An important observation is that the adjacency matrix of lossless expander is a binary matrix with $d$ ones per column in which every submatrix $m\times k$ with $k\le s$ has at least $(1-\theta)dk$ nonzero rows. The next result establishes that lossless expanders provide remarkable guarantees for sparse recovery.

\begin{theorem}{\cite{foucart2013invitation}}
\label{thm:lossless-l1robust}
The adjacency matrix $M \in \{0,1\}^{m\times n}$ of a $(2s,d,\theta)$-lossless expander satisfies the $\ell_1$ robust nullspace property of order $s$ with respect to the $\ell_1$ norm provided that $\theta_{2s} <\frac{1}{6}$. Precisely, for all $v\in \mathbb{R}^n$ and subset $S \subset[n]$ with $|S|\le s$,
\begin{equation*}
    \|v_{S}\|_1 \le \frac{\theta_{2s}}{1-4\theta_{2s}}\|v_{S^c}\|_1 + \frac{1}{(1-4\theta_{2s})d}\|Av\|_1.
\end{equation*}
Moreover, if $G$ is a random graph sampled uniformly at random among all left $d$-regular bipartite graphs with $n$ left nodes and $m$ right vertices, then with probability $1-\varepsilon$, $G$ is a $(s,d,\theta)$-lossless expander provided that
\begin{equation*}
    d = \ceil*{\frac{1}{\theta}\log(\frac{en}{\varepsilon s})} \quad \text{and} \quad m\ge c_{\theta}s\log(\frac{en}{\varepsilon s}).
\end{equation*}
Here $c_{\theta}>0$ is constant that depends only on $\theta$.
\end{theorem}

\section{The quasi-regular lossless expanders}
\label{sec:quasi-regular_expanders}
A closer look at the idea to view a matrix $A\in \{0,1\}^{m\times n}$ formed by i.i.d entries Bernoulli $p$ as an adjacency matrix of a lossless expander, in the sense of \eqref{def:adjancency}, reveals that it is not possible. On average, the matrix $A$ has $mp$ ones per column, but we need to take in account the fluctuations of the degree. On the other hand, the left degree is a binomial random variable and concentrates well around the mean. To overcome this issue, we extend the notion of lossless expanders to quasi-regular lossless expander.

\begin{definition}{(Quasi-regular lossless expanders)}
A bipartite graph $G=G(L,R,E)$ is said to be a $(\delta,d)$-left quasi regular if the left degree of all vertices in $L$ lies on the interval $[(1-\delta)d,(1+\delta)d]$. Moreover, a $(\delta,d)$-left quasi regular is a $(s,d,\delta,\theta)$- quasi regular lossless expander if it satisfies the following expansion property
\begin{equation*}
    |R(J)| \ge (1-\theta)d |J|,
\end{equation*}
for all sets $J$ of left vertices, with cardinality $|J|\le s$. The smallest $\theta$ such that the above property holds is denoted by $\theta_{s}$.
\end{definition}

It should be clear that for $\delta=0$, the quasi-regular lossless expander becomes the lossless expander. In this work, it is more appropriate to think that the quasi-regular lossless expanders are represented by matrices with $(1\pm \delta)d$ ones per column and each $m\times k$ submatrix with $k\le s$ has at least $(1-\theta)dk$ nonzero rows.

The next natural step is to obtain a version of Theorem \ref{thm:lossless-l1robust} for Bernoulli $p$ matrices. We split our analysis in two parts: The first one is purely deterministic, it establishes that the adjacency matrix of the quasi-regular lossless expander satisfies the $\ell_1$ robust nullspace property. The second part is about the random construction, it shows that the Bernoulli $p$ matrix $A$ is the adjacency matrix of a quasi-regular lossless expander with high probability. Let us state the formal result related to the first part.

\begin{theorem}
\label{thm:important1}
The adjacency matrix $A \in \{0,1\}^{m\times n}$ of a $(2s,d,\delta,\theta)$- quasi regular lossless expander satisfies the $\ell_1$ null space property of order $s$ with parameters $\rho = \frac{2\theta_{2s}+6\delta}{1-4\theta_{2s}-13\delta} $ and $\tau = \frac{1}{d(1-4\theta_{2s}-13\delta)}$ with respect to the $\ell_1$ norm provided that the constants $\theta_{2s}$ and $\delta$ satisfy $6\theta_{2s}+ 19\delta < 1$.
\end{theorem}

In particular, for $\delta = 0$, we recover the first part of the statement in Theorem \ref{thm:lossless-l1robust}. The proof is left to the Appendix. We remark that, although the original proof for the deterministic part of Theorem \ref{thm:lossless-l1robust} uses the intuitive idea of $\ell_1$ restricted isometry property, we opt for the proof strategy similarly to \cite{foucart2013invitation} because it goes directly to the nullspace property without passing to any intermediate restricted isometry property and also takes in the account robustness against noise. 

We now provide the main result of this paper. It confirms the intuition that Bernoulli $p$ matrices are a candidate to extend the notion of lossless expanders for non-regular graphs. We remark that such natural extension was possible due to the combinatorial definition of the lossless expanders. The extension of Ramanujan graphs to non-regular graphs is more delicate \cite{bordenave2018nonbacktracking}.

\begin{theorem}
\label{thm:important2}
Let $A \in \{0,1\}^{m\times n}$ with i.i.d entries following a Bernoulli distribution with parameter $p$. There exist an absolute constant $c>0$ such that, for every $0<\theta < \frac 23$, if $m\ge \frac{C}{\delta^2}\frac{\log(n)}{p}$ and $ps\le \frac{2\theta}{2-\theta}$, then, with probability at least $\frac{c}{n^{C-1}}$, the matrix $A$ is the adjacency matrix of a $(s,d,\delta,\theta)$-quasi regular lossless expander with $d=mp$.
\end{theorem}

\begin{proof}
It is sufficient to establish that, with high probability, the matrix $A$ has $(1\pm \delta) d$ ones per column (condition 1) and every $m\times k$ submatrix, with $k\le s$, has at least $(1-\theta)dk$ nonzero rows (condition 2).
\\

\textbf{Condition 1}: Each column has $mp$ ones on average. By Chernoff deviation \cite[Exercise 2.3.5]{vershynin2018high}, the probability that the number of ones deviates from the average more than $\delta mp$ is $O(e^{-\delta^2mp})$. By union bound, the probability that exits a column with number of ones larger than $(1+\delta)mp$ or less than $(1-\delta)mp$ is $O(ne^{-\delta^2 mp}) = O(\frac{1}{n^{C-1}})$. The latter vanishes if $C>1$. 
\\

\textbf{Condition 2}: Now we proceed to the second requirement. For each $k \in [s]$, the number of nonzero rows in a $m\times k$ submatrix follows a Binomial distribution with parameters $m$ and $q:=1-(1-p)^k$. We fix a $m\times k$ submatrix and then the number of nonzero rows has average $mq$. Let $\varepsilon >0$ be a constant to be chosen later, we write
\begin{equation*}
    \begin{split}
        mq(1-\varepsilon)= m(1-(1-p)^k)(1-\varepsilon) &> m(1-e^{-pk})(1-\varepsilon)\\
        &> m(pk-\frac{p^2k^2}{2!})(1-\varepsilon) \ \text{(by Taylor expansion)}\\
        &>m(1-\theta)pk \ \text{for $pk< \frac{2\theta}{2-\theta}<1$}
    \end{split}
\end{equation*}

For the last bound to hold we choose $\varepsilon := \frac{\theta}{2}$. By Chernoff deviation, we get the probability that the number of nonzero rows is less than $(1-\varepsilon)mq$ is $O(e^{-\varepsilon^2mq})$. By union bound, the probability that exists a submatrix $m\times k$ with less than $(1-\theta)dk$ nonzero rows is of order at most
\begin{equation*}
    \sum_{k=1}^s \binom{n}{k} e^{-\varepsilon^2mq}.
\end{equation*}
Now, by Taylor expansion, $mq>m(1-e^{-pk}) > m(pk-\frac{p^2k^2}{2!})>\frac{mpk}{2}$, the last inequality follows from the fact that $pk\le ps< 1$ (recall that $\theta<\frac{2}{3}$). Together with the fact that $mp>C\log n$, we obtain $e^{-\varepsilon^2mq} \le e^{-Ck\log n}$ while $\binom{n}{k} \sim e^{k\log \frac nk}$. We can conclude that the probability above is at most $O(\sum_{k=1}^s \frac{1}{n^{(C-1)k}}) = O(\frac{1}{n^{C-1}})$ and this clearly vanishes for large enough constant $C>0$.
\end{proof}

Combining Theorems \ref{thm:important1} and \ref{thm:important2}, we immediately obtain Theorem \ref{thm:Main}.

\section{Lower Bound \& Phase Transition}
\label{sec:lowerbounds}
In this section, we investigate the optimality of our results. The lower bound presented below suggests that a phase transition for the number of measurements $m$ according to the range of $p$. The proof of such bound is quite simple but the intuition is more delicate. To start, notice that a necessary condition to establish exact recovery guarantees by minimizing the $\ell_0$ norm is that every collection of $2s$ columns must be linearly independent. Similarly, a necessary condition for the $\ell_1$ minimization to work is that every collection of $s$ columns are linearly independent. It is natural to ask the following question: What is the probability that exists a collection of $s$ columns that are linearly dependent? 

It turns out that, for some ranges of $s$, the question above is a delicate problem in the theory of invertibility of random matrices. A line of work \cite{tao_vu_2006,tikhomirov2020singularity,litvak2020singularity} suggests that the phenomenon is local, i.e, the linear dependence is caused by one or two columns, in particular the presence of a zero column. For example, for a square $n\times n$ Bernoulli $p$ matrix, if $p<\frac{\log n}{n}$, then the matrix has a zero column with positive probability, on the other hand if $p\ge(1+c)\frac{\log n}{n}$, where $c>0$ is a constant, then the matrix is invertible with high probability \cite[Corollary 1.3]{basak2021sharp}. We do not use any sophisticated machinery to establish the lower bound, but we borrow the intuition from this line of work. Our lower bound simply estimates the probability to have a zero column.

\begin{proposition}{(Lower bound)}
Let $A \in \{0,1\}^{m\times n}$ be a Bernoulli $p$ matrix. If the matrix $A$ satisfies a nullspace property of order $s$ with high probability, then it is necessary that $m\ge c(s\log\frac{en}{s} + \frac{\log n}{p})$ for an absolute constant $c>0$.
\end{proposition}

\begin{proof}
The lower bound  $m\gtrsim s\log\frac{en}{s}$ is quite standard in the literature (Theorem 10.11 \cite{foucart2013invitation}). Now we evaluate the probability to have a zero column. If the $j$-th column of $A$ is a zero column, then it fails to recover the vector $e_j$, the $j$-th vector of the canonical basis. We write,
\begin{equation*}
\begin{split}
    \mathbb{P}( \exists j \in [n]: Ae_j = 0) &= 1-\mathbb{P}(\forall j\in [n]: Ae_j \neq 0)\\
    &= 1-\mathbb{P}(Ae_1 \neq 0)^n \ \text{(by independence)}\\
    &= 1- (1-(1-p)^m)^n.
    \end{split}
\end{equation*}
If we obtain that, for a certain $p$, $(1-(1-p)^m)^n \rightarrow C<1$, then we get a lower bound on $p$. Notice that $(1-(1-p)^m)^n \le e^{-(1-p)^m n}$. We fix for a moment $m=\frac{c\log n}{p}$, where $c>0$ is a constant to be chosen later. Since $p\le \frac 12$ we have $(1-p)^{1/p} \ge \frac{1}{\sqrt{2}}$ (the function $(1-x)^{1/x}$ is non-increasing from $0$ to $1$), therefore $(1-p)^{c\log n/p} \ge (\frac{1}{\sqrt{2}})^{c\log n} = e^{-c(\log \sqrt{2})\log n} =  e^{-\log n} = \frac{1}{n}$ for $c=\frac{1}{\log \sqrt{2}}$. Finally, we obtain that $(e^{-(1-p)^m})^n \le (e^{-\frac{1}{n}})^n = e^{-1}$. We obtained that $m \ge \frac{1}{\log \sqrt{2}}(\frac{\log n}{p})$ is necessary. We average the two lower bounds to finish the proof.
\end{proof}
We point out that the lower bound and Theorem \ref{thm:Main} show that, in the very sparse regime $p\lesssim\frac{1}{s}$, the main reason that causes the failure of $\ell_1$ relaxation is to have a zero column. In the context of invertibility, this shows that the existence of an absolute constant $C>1$ such that, for $p=C\frac{\log n}{n}$,  every collection of $\frac{n}{C\log n}$ columns of a square Bernoulli $p$ matrix is linearly independent with high probability. The factor $C\log n$ may be removed if we extend the range of $p$ in Theorem \ref{thm:Main} (see the conjecture below), however the main purpose of writing the result in the language of invertibility is to highlight the same local phenomenon and not to prove invertibility statements.

An immediate consequence of the lower bound is that our main result Theorem \ref{thm:Main} is sharp up to absolute constants, but it may admit an extension for larger values of $p$. Since the lower bound has a matching upper bound up to an absolute constant for $p$ constant and also for $p\lesssim\frac{1}{s}$, we conjecture the following result:

\begin{conjecture}
Let $A \in \{0,1\}^{m\times n}$ be a matrix formed by i.i.d entries Bernoulli $p$. There exists constants $C_1,C_2>0$ such that if $m\ge\max\{C_1s\log\frac{en}{s},C_2\frac{\log n}{p}\}$, then $A$ satisfies the $\ell_1$ robust nullspace property of order $s$ for some $\rho <1$ and $\tau>0$. Moreover, we can take any constant $C_2 >1$. 
\end{conjecture}

The "moreover" part of the conjecture is not relevant for compressed sensing, but it would capture the phase transition in \cite[Corollary 1.3]{basak2021sharp}. Notice that, if the conjecture is true, then a natural phase transition occurs:
\begin{equation*}
    m(p) \sim \begin{cases}
    s\log\frac{en}{s}, \ \text{for $ \frac{\log(n)}{s\log \frac{en}{s}}\lesssim p\le \frac{1}{2}$} \\
    \frac{\log n}{p}, \ \text{for  $p\ll \frac{\log(n)}{s\log \frac{en}{s}} $}
    \end{cases}
\end{equation*}
It seems that the best choice of $p$ is $p^{\ast} \sim \frac{\log(n)}{s\log \frac{en}{s}}$ because it balances optimality in terms of measurements with the smallest possible column sparsity of the measurement matrix.
We end this section with some remarks about the conjecture: The small ball method and the standard RIP techniques completely fail in the sparse regime ($p \rightarrow 0$). Moreover, the expander technique used in this manuscript fails in the regime $p\gg \frac{1}{s}$ because the expansion property requires more edges than the maximum number of edges possible. The graph becomes too dense. A natural approach would be to interpolate between $\ell_2$ RIP used in the dense regime and a modified version of $\ell_1$ RIP. Unfortunately, $\ell_p$ RIP fails to give sharp results for $p>1$ and $p\neq 2$ \cite{7534879}.

\section{Applications in nonnegative compressed sensing}
\label{sec:applications}
In this section, we provide a practically relevant application of our main Theorem \ref{thm:Main}. In many applications of compressed sensing, we do not have any knowledge of the noise level. In this sense, noise blind compressed sensing is the topic dedicated to the study of algorithms that perform sparse recovery without using any information about the noise as an input. Arguably, the most common way to tackle this question is to establish the so called quotient property \cite{foucart2013invitation}, however no estimate is known for non-symmetric distributions. Instead, we use the special structure of nonnegative vector recovery to establish a noise blind guarantee for the recovery of nonnegative vectors with a Bernoulli $p$ random matrix. Our analysis is inspired by the work in \cite{kueng2017robust}. To start, we recall the notion of positive orthant
\begin{equation*}
    \mathcal{M}^{+}:= \{A: \exists t \in \mathbb{R}^m \ \text{such that} \ A^{T}t >0\}.
\end{equation*}
The following result establishes that Bernoulli $p$ random matrices obeys the positive orthant condition.

\begin{proposition}{\cite{kueng2017robust}}
\label{Theorem12:Kueng}
Suppose that $A \in \{0,1\}^{m\times n}$ is a Bernoulli $p$ matrix and set 
\begin{equation*}
    w = A^{T}t \in \mathbb{R}^n \quad \text{with} \quad t:=\frac{1}{pm}\mathbf{1}_m \in \mathbb{R}^m.
\end{equation*}
Then, with probability at least $1-ne^{-\frac{3}{8}p(1-p)m}$,
\begin{equation*}
    \max_{i\le n}|\langle e_i,w\rangle| \le \frac{3}{2} \quad \text{and} \quad \min_{i\le n}|\langle e_i,w\rangle| \ge \frac{1}{2}.
\end{equation*}
In particular, $A\in \mathcal{M}^{+}$.
\end{proposition}

We proceed with two easy lemmas, but they are crucial to our result. For notation simplicity, we write
\begin{equation*}
    \kappa(A):= \{\|W\|\|W^{-1}\| \ |\exists t \in \mathbb{R}^m, W=\diag(A^{T}t)>0\}.
\end{equation*}

\vspace{0.2cm}
\begin{lemma}
\label{lemma5-kueng}
Suppose the matrix $A$ satisfies the $\ell_1$ robust nullspace property of order $s$, with parameters $\rho$ and $\tau$, with respect to the $\ell_1$ norm. If $W=\diag(w)$ with $w\ge 0$, then $AW^{-1}$ also satisfies the $\ell_1$ robust nullspace property of order $s$, with parameters $\Tilde{\rho}:=\kappa(W)\rho$ and $\Tilde{\tau}:=\|W\|\tau$, with respect to the $\ell_1$ norm.
\end{lemma}
\begin{proof}
We write
\begin{equation*}
    \|v_{S}\|_1 = \|WW^{-1}v_{s}\|_1\le \sup_{z:\|z\|_1=1} \|Wz\|_1 \|W^{-1}v_{S}\|_1 = \|W\|\|(W^{-1}v)_{S}\|_1.
\end{equation*}
The last equality follows from the fact that $W$ is a diagonal matrix. Now we use the $\ell_1$ robust nullspace property to obtain
\begin{equation*}
\begin{split}
    \|v_{S}\|_1 &\le \|W\|(\rho \|(W^{-1}v)_{S^c}\|_1 + \tau \|AW^{-1}v\|_1) \\
    &\le (\|W\|\|W^{-1}\|\rho)\|v_{S^c}\|_1 + (\tau \|W\|)\|AW^{-1}v\|_1.
    \end{split}
\end{equation*}
In the last inequality, we used the fact that $W^{-1}$ is also a diagonal matrix.
\end{proof}
\begin{lemma}
\label{lemma6-kueng}
Let $A \in \mathbb{R}^{m\times n}$ be a matrix and suppose the existence of a vector $t \in \mathbb{R}^m$ such that $w=A^{T}t$ has nonnegative coordinates. If $W = \diag(w)$, then for all nonnegative vectors $x,z \in \mathbb{R}^n$,
\begin{equation*}
    \|Wz\|_1 - \|Wx\|_1 \le \|t\|_{\infty}\|A(z-x)\|_1.
\end{equation*}
\end{lemma}

\begin{proof}
Observe by construction of $W$,
\begin{equation*}
    \|Wz\|_1 = \langle Wz, \mathbf{1}_n\rangle  = \langle z, W\mathbf{1}_n\rangle = \langle z,\diag(A^{T}t)\mathbf{1}_n\rangle = \langle A^{T}t,z\rangle  = \langle t,Az \rangle.
\end{equation*}
Similarly, the same holds for $\|Wx\|_1$. Therefore, we obtain that
\begin{equation*}
    \|Wz\|_1 - \|Wx\|_1 = \langle t,A(z-x) \rangle \le \|t\|_{\infty}\|A(z-x)\|_1
\end{equation*}
The last step follows from H\"older's inequality.
\end{proof}

We are now in position to state the main result of this section. A similar result for the standard lossless design was recently obtained in \cite{petersen}.

\begin{theorem}
\label{thm:applied}
Let $A\in \{0,1\}^{m\times n}$ be a Bernoulli $p$ matrix. Suppose also that we receive a vector $y=Ax+e$ where $x\in \mathbb{R}^n$ is an unknown vector and $e \in \mathbb{R}^n$ is an unknown noise vector. Then there exists absolute constants $c,C>0$ such that, if $m=\frac{\log n}{p}$ with $p<\frac{1}{Cs}$, with probability at least $1-\frac{c}{n^{C-1}}$, any solution of 
\begin{equation}
\label{eq:l1-LS}
\hat{x}= \arg \min_{z\ge 0}  \|Az - y\|_1
\end{equation}
satisfies 
\begin{equation*}
    \|\hat{x}-x\|_1 \leq C_1\left (\sigma_s(x)_1 + \frac{\|e\|_1}{pm}\right ).
\end{equation*}
Here $C_1>0$ is a constant depending only on $C>0$. Moreover, under the same assumptions, with the same number of measurements and the same probabilistic guarantee,
\begin{equation}
\label{eq:l2-LS}
\hat{x}= \arg \min_{z\ge 0}  \|Az - y\|_2
\end{equation}
satisfies 
\begin{equation*}
    \|\hat{x}-x\|_1 \leq C_2\left(\sigma_s(x)_1 + \frac{\|e\|_2}{p\sqrt{m}}\right).
\end{equation*}
Here $C_2>0$ is a constant depending only on $C>0$.
\end{theorem}

Our theorem should be compared with \cite[Theorem 3]{kueng2017robust}. For simplicity, we assume $p\sim \frac{1}{s}$. All results are about noise blind nonnegative signal recovery, i.e, under the assumption that the vector to be recovered is nonnegative neither \eqref{eq:l1-LS} nor \eqref{eq:l2-LS} require knowledge of the noise level as an input. Besides, \cite[Theorem 3]{kueng2017robust} provides guarantees for the optimization program \eqref{eq:l2-LS} with the error bound scaling with $\frac{1}{m}$ instead of $\frac{1}{\sqrt{m}}$, however the error is measured in the $\ell_2$ norm instead of the $\ell_1$ norm. Theorem \ref{thm:applied} requires slightly more measurements than the optimal, the difference between $\log \frac{en}{s}$ and $\log n$. Again, they are of the same order in the wide range $s=O(n^{\beta})$ for some constant $0<\beta <1$. On the other hand, the measurement matrix $A$ in Theorem \ref{thm:applied} is much sparser as $p\sim \frac{1}{s}$ instead of just being a constant. Finally, we hope that some sub-linear time algorithms designed for the lossless expanders can be easily adapted to the Bernoulli $p$ matrix because the latter is a quasi-regular expander. We do not pursue this direction here. We now proceed to the proof of Theorem \ref{thm:applied}.

\begin{proof}
Thanks to Proposition \ref{Theorem12:Kueng}, we know that if we choose $t=\frac{1}{pm}\mathbf{1}_n$, then with probability at least $1-ne^{-\frac{3}{8}p(1-p)m}$, $w=A^{T}t >0$. Under this event, the matrix $W:=\diag(w)$ is invertible with condition number $\kappa:=\|W\|\|W^{-1}\| \sim 1$. By Theorem \ref{thm:Main}, with probability at least $1-O(\frac{1}{n^{C-1}})$, the matrix $A$ satisfies the $\ell_1$ robust nullspace property of order $s$ with parameters $\rho \sim 1$ and $\tau \sim \frac{1}{\log n}$ with respect to the $\ell_1$ norm. For the rest of the proof, we assume that both events holds. It is easy to see that for a sufficient large constants $C,c>0$, with probability at least $1-\frac{c}{n^{C-1}}$, both events holds simultaneously. Now, we apply Lemma \ref{lemma5-kueng} to obtain that $AW^{-1}$ also satisfies robust NSP of order $s$ with parameters $\Tilde{\rho} = \kappa \rho \sim 1$ and $\Tilde{\tau}\sim \|W\|\frac{1}{pm} \sim \frac{1}{pm}$. Now we apply Theorem \ref{Thm:4.25-Fourcart/Rauhut} to obtain that, for vectors $Wz$ and $Wx$, the following holds
\begin{equation*}
\begin{split}
\|W(z-x)\|_1 &\lesssim \|Wz\|_1-\|Wx\|_1 + \sigma_s(Wx)_1 + \|W\| \tau \|A(x-z)\|_1\\
& \lesssim \|t\|_{\infty}\|A(z-x)\|_1 + \|W\|\sigma_s(x)_1 + \|W\| \tau \|A(x-z)\|_1\\
& = \|W\|\sigma_s(x)_1 +\|A(z-x)\|_1 (\frac{1}{pm}+\|W\|\tau)
\end{split}
\end{equation*}
The last inequality follows from Lemma \ref{lemma6-kueng} and the fact that $\sigma_s(Wx)_1\le \|W\|\sigma_{s}(x)_1$. Now we obtain
\begin{equation*}
\begin{split}
\|x-z\|_1 & \le \|W^{-1}\|\|W(z-x)\|_1 \\
& \lesssim \|W^{-1}\|(\|W\|\sigma_s(x)_1 +\|A(z-x)\|_1 (\frac{1}{pm}+\|W\|\tau))\\
&\lesssim \kappa \sigma_s(x)_1 + \|A(z-x)\|_1 \left(\frac{\|W^{-1}\|}{pm}+\kappa\tau\right) \\
&\lesssim \sigma_s(x)_1 + \frac{1}{pm}\|A(z-x)\|_1.
\end{split}
\end{equation*}
To get the first result, choose $z=z^{\ast}$ to be the minimizer of the optimization program \eqref{eq:l1-LS}, clearly $\|A(z^{\ast}-x)\|_1 \le \|Az^{\ast}-y\|_1 + \|Ax-y\|_1 \le 2\|e\|_1$. The second results follows from the same arguments with an extra application of Cauchy-Schwarz inequality.

\end{proof}

\section{Conclusion}
\label{sec:conclusion}
In this work, we analysed the robust null space property of a Bernoulli $p$ matrix in the practically relevant regime when $p$ vanishes. We derived sharp bounds for the minimal number of measurements required to establish robust nullspace property of Bernoulli $p$ matrices for a wide range of sparsity levels $s$. We formulated a precise conjecture that reveals a phase transition not observed in the previous works. Moreover, we provide applications of our main result in nonnegative compressed sensing.

\section{Acknowledgement}
The author would like to thank Afonso Bandeira, Hendrik Peterson and Nikita Zhivotovsky for helpful discussions. The author is grateful to Geoffrey Chinot and Sara van de Geer for a careful revision of the manuscript.

\section{Appendix}
\label{sec:appendix}
The appendix is dedicated to the proof of Theorem \ref{thm:important1}. In this section, we assume that the implicit graph is an $(s,d,\delta,\theta)$ quasi regular lossless expander. To start, we present two lemmas that reflect the intuition that a lossless expander posses a small number of collisions. 

\begin{lemma}
\label{lemma_for_E(J;K)}
If $J$ and $K$ are disjoint sets of left vertices satisfying $|J|+|K| \le s$, then the set 
\begin{equation*}
    E(J;K):=\{\overline{ji} \in E(K) \ \text{with} \ i\in R(J)\}
\end{equation*}
has cardinality at most $(\theta_s+3\delta)ds$.
\end{lemma}

\begin{proof}
We partition the set of edges emanating from $E_0:=J\cup K$, in three distinct sets:
\begin{enumerate}
    \item The set $E_1$ of edges emanating from $J$
    \item The set $E_2$ of edges emanating from $K$ and whose right vertices not connected to any left vertex in $J$
    \item The set $E_3$ of edges emanating from $K$ and whose right vertices are also connected to the left vertices in $J$
\end{enumerate}
We need to bound $E_3$. We know that $|E_0|\le (1 + \delta)d(|J|+|K|)$ and $|E_1|\ge (1 - \delta)d|J|$ because of the quasi-regularity of the graph. We get $|E_3|=|E_0|-|E_1|-|E_2|\le 2\delta d|J| +(1+\delta)d|K| -|E_2|$. We proceed to bound $E_2$. Observe that each right vertex $i \in R(K)/R(J)$ gives rise to at least one edge emanating from $K$ whose right vertex is not connected to any left vertex in $J$. We can write $|E_2|\ge |R(K)/R(J)|=|R(J\cup K)|-|R(J)|$. Since $|R(J)|\le (1+\delta)d|J|$ and $|R(J\cup K)| \ge (1-\theta_s)d(|J|+|K|)$ we can bound
\begin{equation*}
    |E_2|\ge (1-\theta_s)d(|J|+|K|) - (1+\delta)d|J| = |J|d(-\theta_s-\delta) + (1-\theta_s)d|K|.
\end{equation*}
Now we use this bound to estimate $|E_3|$,
\begin{equation*}
    |E_3|\le 2\delta d|J|+|J|d(\theta_s+\delta)+(1+\delta)d|K|-(1-\theta_s)d|K| = |J|d(\theta_s+3\delta) + |K|d(\delta+\theta_s).
\end{equation*}
Recall that $|J|+|K|\le s$, so we get $|E_3|\le sd(\theta_s+3\delta)$.
\end{proof}

\begin{lemma} 
\label{simple_lemma}
Let $S$ be a set of $s$ left vertices. For each $i\in R(S)$, if $l(i) \in S$ denotes fixed left vertex connected to $i$, then
\begin{equation*}
    E^{\ast}(S):=\{\overline{ji} \in E(S): j\neq l(i)\},
\end{equation*}
has cardinality at most $(\theta_s+\delta) ds$.
\end{lemma}

\begin{proof}
The set $E(S)$ of the edges emanating from $S$ can be partitioned into two sets $E^{\ast}(S)$ and $E_c^{\ast}(S):=\{\overline{l(i)i}, i\in R(S)\}$. Since $|E(S)|\le (1+\delta)ds$ and $E_c^{\ast}(S) \ge (1-\theta_s)ds$, we get the desired bound.
\end{proof}

The next two lemmas are about technical properties of lossless expanders.

\begin{lemma}
If $S$ and $T$ are two disjoint subsets of $[n]$ and if $x\in \mathbb{R}^n$, then
\begin{equation*}
    \|(Ax_S)_{R(T)}\|_1 \le (\theta_{s+t}+3\delta)d(s+t)\|x_S\|_{\infty},
\end{equation*}
where $s=|S|$ and $t=|T|$.
\end{lemma}
\begin{proof}
We estimate the term $\|(Ax)_{R(T)}\|_1 $ as follows
\begin{equation*}
    \begin{split}
        \|(Ax)_{R(T)}\|_1 &= \sum_{i\in R(T)}|(Ax_S)_i| = \sum_{i\in R(T)}\mathds{1}_{i\in R(T)}|\sum_{j \in S}A_{ij}x_j|\\
        & \le \sum_{i=1}^m \mathds{1}_{i\in R(T)} \sum_{j\in S}\mathds{1}_{\{\overline{ji} \in E\}} |x_j| = \sum_{\overline{ji} \in E(S;T)}|x_j|\\
        &\le |E(S;T)|\|x_{S}\|_{\infty}.
    \end{split}
\end{equation*}
We now apply Lemma \ref{lemma_for_E(J;K)} to conclude the proof.
\end{proof}

\begin{lemma}
\label{last_lemma}
Given a $s$-sparse vector $w \in \mathbb{R}^n$, let $w^{\ast}$ be defined by $w^{\ast}_{i}:=w_{l(i)}$ where $l(i):=\arg\max_{j\in [n]}\{w_{ij}, \ \overline{ji}\in E\}$. Then,
\begin{equation*}
    \|Aw-w^{\ast}\|_{1} \le (\theta_{s}+\delta)d\|w\|_1.
\end{equation*}
\end{lemma}

\begin{proof}
Without loss of generality, assume that the entries of $w$ are ordered in a non-increasing order. The edge $\overline{l(i)i}$ can be though of as the first edge arriving at the right vertex $i$. We denote the support of the vector $w$ by $S$ and write
\begin{equation*}
    (Aw-w^{\ast})_i = \sum_{j=1}^n A_{ij}w_j - w_{l(i)} = \sum_{j\in S}\mathds{1}_{\{\overline{ji}\in E \ \text{and} \ j\neq l(i)\}}w_j.
\end{equation*}
Now we write,
\begin{equation*}
    \|Aw - w^{\ast}\|_1 = \sum_{i=1}^m |\sum_{j \in S}\mathds{1}_{\overline{ji}\in E; j\neq l(i)}w_j| \le \sum_{j\in S}(\sum_{i=1}^m \mathds{1}_{\overline{ji}\in E; j\neq l(i)})|w_j| := \sum_{j=1}^s c_j |w_j|
\end{equation*}
Thus we get $\|Aw-w^{\ast}\|_1 = \sum_{j=1}^s c_jw_j$, where $c_j:= \sum_{i=1}^m \mathds{1}_{\{\overline{ji}\in E \ \text{and} \ j\neq l(i)\}}$. For all $k\in [s]$, we have
\begin{equation*}
    C_k:=\sum_{j=1}^k c_j = |\{\overline{ji} \in E([k]), j\neq l(i)\}| \le (\theta_s+\delta)dk,
\end{equation*}
where the last inequality follows from Lemma \ref{simple_lemma}. We now perform summation by parts (assuming $C_0:=0$),
\begin{equation*}
    \sum_{j=1}^s c_j w_j \le \sum_{j=1}^s C_j|w_j|-C_{j-1}|w_{j}|=\sum_{j=1}^{s-1}C_j(|w_j|-|w_{j+1}|) + |C_s w_s|
\end{equation*}
Since $|w_j|-|w_{j+1}|\ge 0$ due to the non-increasing order, we can plug the bound for $C_k$ into the summation above to get
\begin{equation*}
    \sum_{j=1}^s c_jw_j \le \sum_{j=1}^s (\theta_s+\delta) d|w_j|=(\theta_s+\delta)d\|w\|_1.
\end{equation*}
\end{proof}

We are now in position to prove Theorem \ref{thm:important1}. 
\begin{proof}{(of Theorem \ref{thm:important1})}
Let $v\in \mathbb{R}^n $ be a fixed vector. We define $S_0$ to be the set of $s$ largest entries of $v$ in absolute value, $S_1$ to be the set of $s$ largest entries not in $S_0$. The sets $S_2,S_3,\ldots,S_{\lceil\frac{n}{s}\rceil}$ are defined analogously. Now, we write

\begin{equation}
\label{important_equality}
\begin{split}
   (1-\delta) d\|v_{S_0}\|_1 = (1-\delta)d\sum_{j \in S_0}|v_j| &\le \sum_{\overline{ji}\in E(S_0)} |v_j| = \sum_{i \in R(S_0)}\sum_{\{j\in S_0,\overline{ji}\in E\}}|v_j|\\
    &= \sum_{i \in R(S_0)} |v_{l(i)}| + \sum_{i\in R(S_0)}\sum_{\{j\in S_0 \backslash \{l(i)\},\overline{ji}\in E\}}|v_j|.
    \end{split}
\end{equation}
Moreover, for $i\in R(S_0)$,
\begin{equation*}
    (Av)_i =\sum_{j\in [n]} A_{ij}v_j = v_{l(i)} + \sum_{\{j\in S_0 \backslash l(i),\overline{ji}\in E \}} v_j + \sum_{k\ge 1}\sum_{\{j\in S_k,\overline{ji}\in E \}} v_j.
\end{equation*}
It follows that
\begin{equation*}
    |v_{l(i)}|\le \sum_{\{j\in S_0/l(i),\overline{ji}\in E \}} v_j + \sum_{k\ge 1}\sum_{\{j\in S_k,\overline{ji}\in E \}} v_j + |(Av)_i|
\end{equation*}
Summing over all $i\in R(S_0)$ and using equation \eqref{important_equality},
\begin{equation*}
    (1-\delta)d\|v_{s_0}\|_1 \le \sum_{i\in R(S_0)}\left (2\sum_{\{j\in S_0/l(i),\overline{ji}\in E \}} v_j + \sum_{k\ge 1}\sum_{\{j\in S_k,\overline{ji}\in E \}} v_j \right ) + \|Av\|_1.
\end{equation*}
In order to handle the first term in the right hand side, we apply Lemma \ref{last_lemma} with $w=v_{S_0}$ to obtain that 
\begin{equation*}
    2\sum_{i\in R(S_0)}\sum_{\{j\in S_0/l(i),\overline{ji}\in E \}} v_j = 2\|Aw-w^{\ast}\|_1 \le 2(\theta_s+\delta)d\|v_{S_0}\|_1.
\end{equation*}

\noindent For the second term, Lemma \ref{lemma_for_E(J;K)} implies that

\begin{equation*}
\begin{split}
    \sum_{k\ge 1}\sum_{i\in R(S_0)}\sum_{j\in S_k,\overline{ji}\in E}|v_j|&=\sum_{k\ge 1}\|(Av_{S_k})_{R(S_0)}\|_1 \le \sum_{k\ge 1}2(\theta_{2s}+3\delta)ds\|v_{S_k}\|_{\infty}\\
    &\le 2(\theta_{2s}+3\delta_2)d\sum_{k\ge 1}\|v_{S_{k-1}}\|_1= 2(\theta_{2s}+3\delta)d\|v\|_1.
    \end{split}
\end{equation*}
Where in the last inequality we used that $s\|v_{S_k}\|_{\infty}\le \|v_{S_{k-1}}\|_1 $. Putting the bounds together,
\begin{equation*}
\begin{split}
    (1-\delta)d\|v_{S_0}\|_1 &\le 2(\theta_s+ \delta)d \|v_{S_0}\|_1 + 2(\theta_{2s}+3\delta)d\|v\|_1 + \|Av\|_1\\
    &\le 4(\theta_{2s}+3\delta)d\|v_{S_0}\|_1 + 2(\theta_{2s}+3\delta)d\|v_{{S_0}^c}\|_1 + \|Av\|_1.
\end{split}
\end{equation*}
Rearranging the inequality we conclude the proof.
\end{proof}

\bibliographystyle{abbrv}
\bibliography{References.bib}

\end{document}